\documentclass[12pt]{article}

\usepackage[latin1]{inputenc} 
\usepackage{amssymb} 
\usepackage{amsmath}
\usepackage{amsthm}
\usepackage{mathtools}
\usepackage{txfonts}
\usepackage{latexsym}
\usepackage{subfigure}
\usepackage{graphicx}
\usepackage{bm}
\usepackage{bbm}
\usepackage{overpic}
\usepackage[normalem]{ulem}
\usepackage{url}
\usepackage{soul}
\usepackage{exscale}
\usepackage{amsfonts}
\usepackage[usenames,dvipsnames]{xcolor} 
\usepackage{hyperref}
\usepackage{verbatim}
\usepackage[usenames,dvipsnames]{xcolor}
\usepackage{comment}
\usepackage{epsfig}
\usepackage[algosection]{algorithm2e}

\usepackage{tikz}
\usepackage{mhequ} 
\usetikzlibrary{snakes}
\usetikzlibrary{decorations}
\usetikzlibrary{positioning}
\usetikzlibrary{shapes}
\usetikzlibrary{external}
\newtheorem{definition}{Definition}[section]
\newtheorem{lemma}{Lemma}[section]
\newtheorem{remark}{Remark}[section]

\newtheorem{theorem}{Theorem}[section]
\newtheorem{corollary}{Corollary}[section]

\numberwithin{equation}{section}
\numberwithin{figure}{section}

\usepackage{tikz}
\newcommand{\tkz}{\tikzexternaldisable}
\usepackage{mhequ} 
\usetikzlibrary{snakes}
\usetikzlibrary{decorations}
\usetikzlibrary{positioning}
\usetikzlibrary{shapes}
\usetikzlibrary{external}
\usetikzlibrary{snakes}
\usetikzlibrary{decorations}
\usetikzlibrary{positioning}
\usetikzlibrary{shapes}
\usetikzlibrary{external}
\makeatletter
\pgfdeclareshape{crosscircle}
{
  \inheritsavedanchors[from=circle] 
  \inheritanchorborder[from=circle]
  \inheritanchor[from=circle]{north}
  \inheritanchor[from=circle]{north west}
  \inheritanchor[from=circle]{north east}
  \inheritanchor[from=circle]{center}
  \inheritanchor[from=circle]{west}
  \inheritanchor[from=circle]{east}
  \inheritanchor[from=circle]{mid}
  \inheritanchor[from=circle]{mid west}
  \inheritanchor[from=circle]{mid east}
  \inheritanchor[from=circle]{base}
  \inheritanchor[from=circle]{base west}
  \inheritanchor[from=circle]{base east}
  \inheritanchor[from=circle]{south}
  \inheritanchor[from=circle]{south west}
  \inheritanchor[from=circle]{south east}
  \inheritbackgroundpath[from=circle]
  \foregroundpath{
    \centerpoint%
    \pgf@xc=\pgf@x%
    \pgf@yc=\pgf@y%
    \pgfutil@tempdima=\radius%
    \pgfmathsetlength{\pgf@xb}{\pgfkeysvalueof{/pgf/outer xsep}}%
    \pgfmathsetlength{\pgf@yb}{\pgfkeysvalueof{/pgf/outer ysep}}%
    \ifdim\pgf@xb<\pgf@yb%
      \advance\pgfutil@tempdima by-\pgf@yb%
    \else%
      \advance\pgfutil@tempdima by-\pgf@xb%
    \fi%
    \pgfpathmoveto{\pgfpointadd{\pgfqpoint{\pgf@xc}{\pgf@yc}}{\pgfqpoint{-0.707107\pgfutil@tempdima}{-0.707107\pgfutil@tempdima}}}
    \pgfpathlineto{\pgfpointadd{\pgfqpoint{\pgf@xc}{\pgf@yc}}{\pgfqpoint{0.707107\pgfutil@tempdima}{0.707107\pgfutil@tempdima}}}
    \pgfpathmoveto{\pgfpointadd{\pgfqpoint{\pgf@xc}{\pgf@yc}}{\pgfqpoint{-0.707107\pgfutil@tempdima}{0.707107\pgfutil@tempdima}}}
    \pgfpathlineto{\pgfpointadd{\pgfqpoint{\pgf@xc}{\pgf@yc}}{\pgfqpoint{0.707107\pgfutil@tempdima}{-0.707107\pgfutil@tempdima}}}
  }
}
\makeatother

\def\S{\tikz[baseline=-2.8,scale=0.15]{\node[T1] {};}} 
\def\U{\tikz[baseline=-2.8,scale=0.15]{\node[T2] {};}} 
\def\A{\tikz[baseline=-2.8,scale=0.15]{\node[A] {};}} 
\def\X{\tikz[baseline=-2.8,scale=0.15]{\node[X] {};}} 
\def\Zdotred{\tikz[baseline=-2.8,scale=0.15]{\node[Z] {};}} 
\def\M{\tikz[baseline=-2.8,scale=0.15]{\node[M] {};}} 

\def\AAd{\tikz[baseline=-1,scale=0.15]{\draw (-1,1) node[A] {} -- (0,0) node[not] {} -- (1,1) node[A] {};}} 
\def\XXd{\tikz[baseline=-1,scale=0.15]{\draw (-1,1) node[X] {} -- (0,0) node[not] {} -- (1,1) node[X] {};}} 
\def\ZZd{\tikz[baseline=-1,scale=0.15]{\draw (-1,1) node[Z] {} -- (0,0) node[not] {} -- (1,1) node[Z] {};}} 
\def\MMd{\tikz[baseline=-1,scale=0.15]{\draw (-1,1) node[M] {} -- (0,0) node[not] {} -- (1,1) node[M] {};}} 

\def\T1T2d{\tikz[baseline=-1,scale=0.15]{\draw (-1,1) node[T1] {} -- (0,0) node[not] {} -- (1,1) node[T2] {};}} %
 \def\XZd{\tikz[baseline=-1,scale=0.15]{\draw (-1,1) node[X] {} -- (0,0) node[not] {} -- (1,1) node[Z] {};}} 
\def\XMd{\tikz[baseline=-1,scale=0.15]{\draw (-1,1) node[X] {} -- (0,0) node[not] {} -- (1,1) node[M] {};}} 
\def\MXd{\tikz[baseline=-1,scale=0.15]{\draw (-1,1) node[M] {} -- (0,0) node[not] {} -- (1,1) node[X] {};}}

\def\AAdAAdd{\tikz[baseline=-1,scale=0.15]{\draw (0,0) node[not] {} -- (-1,1) node[not] {};
\draw (0,0) -- (1,1) node[not] {};
\draw (-1,1) -- (-1.5,2.5) node[A] {};
\draw (-1,1) -- (-0.5,2.5) node[A] {};
\draw (1,1) -- (0.5,2.5) node[A] {};
\draw (1,1) -- (1.5,2.5) node[A] {};}}

\def\AAdAdAd{\tikz[baseline=-1,scale=0.15]{\draw (0,0) node[not] {} -- (-1,1) node[not] {}
-- (-2,2) node[not]{} -- (-3,3) node[A]  {};
\draw (0,0) -- (1,1) node[A] {};
\draw (-1,1) -- (0,2) node[A] {};
\draw (-2,2) -- (-1,3) node[A] {};}}

\def\AAdAd{\tikz[baseline=-1,scale=0.15]{\draw (0,2) node[A] {} -- (-1,1) ; \draw (-2,2)  node[A] {} -- (-1,1) ; \draw (-1,1)  node[not] {} -- (0,0); 
 \draw (0,0) node[not] {}  -- (1,1) node[A] {};}}

\def\MXdXd{\tikz[baseline=-1,scale=0.15]{
\draw (0,0) node[not] {} -- (-1,1) node[not] {}
-- (-2,2) node[M]{} ;
\draw (0,0) -- (1,1) node[X] {};
\draw (-1,1) -- (0,2) node[X] {};
}}

\def\XMdMd{\tikz[baseline=-1,scale=0.15]{
\draw (0,0) node[not] {} -- (-1,1) node[not] {}
-- (-2,2) node[X]{} ;
\draw (0,0) -- (1,1) node[M] {};
\draw (-1,1) -- (0,2) node[M] {};
}}

\def\MMdXd{\tikz[baseline=-1,scale=0.15]{
\draw (0,0) node[not] {} -- (-1,1) node[not] {}
-- (-2,2) node[M]{} ;
\draw (0,0) -- (1,1) node[X] {};
\draw (-1,1) -- (0,2) node[M] {};
}}

\def\MXdXdXd{\tikz[baseline=-1,scale=0.15]{
\draw (0,0) node[not] {} -- (-1,1) node[not] {}
-- (-2,2) node[not]{}  -- (-3,3) node[M]{};
\draw (0,0) -- (1,1) node[X] {};
\draw (-1,1) -- (0,2) node[X] {};
\draw (-2,2) -- (-1,3) node[X] {};
}}

\def\AAdAddAAd{\tikz[baseline=-1,scale=0.15]{
\draw (0,0) node[not] {} -- (-1,1) node[not] {}
-- (-2,2) node[not]{} -- (-3,3) node[A]  {};
\draw (-1,1) -- (-.5,2) node[A] {};
\draw (-2,2) -- (-1,3) node[A] {};
\draw (0,0) -- (1,1) node[not] {};
\draw (1,1) -- (2,2) node[A] {};
\draw (1,1) -- (.5,2) node[A] {};
}}

\def\AAdAAddAd{\tikz[baseline=-1,scale=0.15]{
\draw (1,0) node[not] {} -- (0,1) node[not] {};
\draw (1,0) node[not] {} -- (2,1) node[A] {};
\draw (0,1) node[not] {} -- (-1,2) node[not] {};
\draw (0,1) -- (1,2) node[not] {};
\draw (-1,2) -- (-1.5,3.5) node[A] {};
\draw (-1,2) -- (-0.5,3.5) node[A] {};
\draw (1,2) -- (0.5,3.5) node[A] {};
\draw (1,2) -- (1.5,3.5) node[A] {};}}

\def\AAdAdAdAd{\tikz[baseline=-1,scale=0.15]{
\draw (0,0) node[not] {} -- (-1,1) node[not] {}
-- (-2,2) node[not]{} -- (-3,3) node[not]  {} 
-- (-4,4) node[A]{};
\draw (0,0) -- (1,1) node[A] {};
\draw (-1,1) -- (0,2) node[A] {};
\draw (-2,2) -- (-1,3) node[A] {};
\draw (-3,3) -- (-2,4) node[A] {};}}

\colorlet{symbols}{blue!90!black}
\colorlet{testcolor}{green!60!black}
\colorlet{connection}{red!30!black}

\tikzset{
root/.style={circle,fill=black!50,inner sep=0pt, minimum size=3mm},
        dot/.style={circle,fill=black,inner sep=0pt, minimum size=1.5mm},
        T1/.style={very thin,circle,fill=BrickRed!100,draw=black,inner sep=0pt,minimum size=1.2mm},         		T2/.style={very thin,circle,fill=BurntOrange!100,draw=black,inner sep=0pt,minimum size=1.2mm},         	A/.style={very thin,circle,fill=BlueGreen!100,draw=black,inner sep=0pt,minimum size=1.2mm},         		X/.style={very thin,circle,fill=Gray!80,draw=black,inner sep=0pt,minimum size=1.2mm},         			Z/.style={very thin,circle,fill=RubineRed!100,draw=black,inner sep=0pt,minimum size=1.2mm},         		M/.style={very thin,circle,fill=YellowOrange!100,draw=black,inner sep=0pt,minimum size=1.2mm},         
	not/.style={thin,circle,fill=symbols,draw=connection,fill=connection,inner sep=0pt,minimum size=0.35mm},
	>=stealth,
        }

\newcommand{\R}{\mathbb{R}}
 
\newcommand{\BS}{\rm BS}

\newcommand{\p}{\partial}
\newcommand{\var}{{\rm var}}
\newcommand{\cov}{{\rm cov}}
\newcommand{\beas}{\begin{eqnarray*}}
\newcommand{\eeas}{\end{eqnarray*}}
\newcommand{\bea}{\begin{eqnarray}}
\newcommand{\eea}{\end{eqnarray}}
\newcommand{\tmop}{\end{eqnarray}}
\newcommand{\ben}{\begin{enumerate}}
\newcommand{\een}{\end{enumerate}}

\newcommand{\ui}{\mathrm{i}}
\newcommand{\e}{\mathrm{e}}

\newcommand{\dm}{\diamond}

\newcommand{\cS}{\mathcal{S}}
\newcommand{\cL}{\mathcal{L}}
\newcommand{\cF}{\mathcal{F}}

\newcommand{\mF}{\mathbb{F}}
\newcommand{\mT}{\mathbb{T}}
\newcommand{\mK}{\mathbb{K}}

\newcommand{\cG}{\mathcal{G}}

\newcommand{\cC}{\mathcal{C}}
\newcommand{\D}{\mathrm{d}}
\newcommand{\cE}{\mathcal{E}}

\newcommand{\RR}{\mathbb{R}}

\newcommand{\mG}{\mathbb{G}}
\newcommand{\tmF}{\tilde {\mathbb{F}}}

\newcommand{\cO}{\mathcal{O}}

\newcommand{\bi}{\begin{itemize}}
\newcommand{\ei}{\end{itemize}}
\newcommand{\beq}{\begin{equation}}
\newcommand{\eeq}{\end{equation}}
\newcommand{\bv}{\begin{verbatim}}
\newcommand{\ev}{\end{verbatim}}

\newcommand{\EE}{\mathbb{E}}
\newcommand{\ee}[1]{\ensuremath{\mathbb{E}\left[{#1}\right]}}
\newcommand{\eef}[1]{\ensuremath{\mathbb{E}\left[\left.{#1}\right|\cF_t\right]}}

\newcommand{\cEE}[1]{\cE\left({#1}\right)}

\newcommand{\eet}[1]{\mathbb{E}_{t}\left[#1\right]}
\newcommand{\angl}[1]{\langle{#1}\rangle}

\begin{document}
\title{\bf Diamonds and forward variance models}

\author{Peter Friz, TU and WIAS Berlin.\\{\tt friz@math.tu-berlin.de}
\\$~~$\\Jim Gatheral, Baruch College, CUNY,\\ {\tt jim.gatheral@baruch.cuny.edu}}

\date{April 22, 2022}

\maketitle\thispagestyle{empty}

\begin{abstract}
In this non-technical introduction to diamond trees and forests, we focus on their application to computation in stochastic volatility models written in forward variance form, rough volatility models in particular.
\end{abstract}

\section{Diamond trees and forests}

From~\cite{alos2020exponentiation, friz2022forests}, we have the following definition of the diamond product.
\begin{definition} \label{def:diamond}
Given two continuous semimartingales $A,B$ with integrable covariation process $\langle A , B \rangle$, the diamond product is defined by
$$
     (A \dm B)_t (T) := \eef{\langle A , B \rangle_{t,T}} = \eef{\langle A , B \rangle_{T}} -\langle A , B \rangle_t \; .
$$
\end{definition} 
 \noindent  We shall assume that all martingales admit a continuous version, it is then clear that $A \dm B$ is another continuous semimartingale (but in general, due to the covariation term $\langle A , B \rangle$, not a martingale). The diamond product has the following properties.

\bi
\item It is commutative: $A \dm B = B \dm A$.
\item It is non-associative: $(A \dm B) \dm C \neq A \dm (B  \dm C)$.
\item $A \dm B$ depends only on the respective local martingale parts of $A$ and $B$.
\ei
Diamond products of (sufficiently integrable, continuous) semimartingales 
naturally lead to binary ``diamond'' trees 
such as $(A \dm B) \dm C$.  
Diagrammatically, the diamond product of two trees $\mT_1$ and $\mT_2$ is represented by {\it root joining}, 
\tkz
$$
\mT_1 \dm \mT_2 = \T1T2d,
$$
where the two binary trees $\mT_1$ and $\mT_2$ are represented as the single leaves $\S$ and $\U$. We regard linear combinations of diamond trees, as seen in~\eqref{some_forests} below, as {\em forests}.   This tree formalism is extremely convenient when it comes to doing explicit computations as we will see later on.

The following theorem, adapted from~\cite{friz2022forests}, expresses the cumulant generating function of any continuous semimartingale and its quadratic variation as a sum of forests of diamond trees.

\begin{theorem}
\label{thm:mainintroG}
Let $Y_T$ be a real-valued, $\cF_T$-measurable random variable with associated martingale $Y_t=\eet{Y_T}$.
Under natural integrability conditions, with $a,b$ small enough, there is a.s. convergence of
\begin{equation}  \label{eq:AsymCG}
\log \eef{\e^{a Y_T + b \langle Y \rangle_T}} = 
a Y_t + b \langle Y \rangle_t 
+  \sum_{k\ge2}   \mG_t^k (T),
  \end{equation}
where 
\vspace{-20pt}
\bea \label{eq:Grec}
\mG^2  &=& \left( \frac{1}{2} a^2+ b\right)
(Y \dm Y)_t(T)  ,
\nonumber \\
 \mG^k &=& \frac{1}{2} \sum_{j = 2}^{k -2} \mG^{k - j} \dm \mG^j +  (a \,Y  \dm \mG^{k - 1}) \text{ for } k>2.
\eea
For the multivariate case, $Y = (Y^1,...,Y^d)$ and $a \in \R^d, b \in \R^{d \times d} $, we replace $a Y$ and $b \langle Y \rangle$ by $\sum_i a_i Y^i $ and $ \sum_{i,j }b_{ij} \langle Y^i, Y^j \rangle$, and further have $\mG^2 = \sum_{i,j } \left( \frac{1}{2}a_i a_j + b_{ij }\right)
(Y^i \dm Y^j)_t(T)$.

\end{theorem}

\begin{proof}[Proof of Theorem~\ref{thm:mainintroG}]
(Sketch) For a generic (continuous) semimartingale $Z$, sufficiently integrable, let
\[
\Lambda_t^T = \log \eet{\e^{Z_{t,T}}}.
\]
Then,  noting that $\Lambda_T^T=0$,
\[
\eet{\e^{Z_T}} = \eet{\e^{Z_T+\Lambda^T_T}} = \e^{Z_t + \Lambda_t^T}. 
\]
The stochastic logarithm $\cL\left(\EE_\bullet(Z_T)\right) = Z+ \Lambda^T +\tfrac12 \angl{Z+ \Lambda^T}$ is a martingale.  Thus, 
\bea
 \Lambda_t^T &=& \eet{Z_{t,T}  +\tfrac12 \angl{Z+ \Lambda^T}_{t,T} } \nonumber\\
 &=& \eet{Z_{t,T}} + \tfrac12( (Z+ \Lambda^T)\dm(Z+ \Lambda^T))_t(T).
 \label{eq:Master}
\eea
Now with\footnote{Recall that terms of bounded variation such as $\angl{Y}$ do not contribute to diamond products.} 
$Z = \epsilon a Y + \epsilon^2\,b \angl{Y}$ we get
\beas
    \Lambda_t^T (\epsilon)  =\epsilon a\, \eet{Y_{t,T}} +\epsilon^2 \,b\, (Y \dm Y)_t(T) + 
     { \tfrac{1}{2} \left(\epsilon a Y  + \Lambda_t^T (\epsilon)\right)^{\dm 2} _t(T)  }  \;.
\eeas 
Put $\Lambda_t^T (\epsilon) =\epsilon^{2} \mG^2_t + \epsilon^{3} \mG^3_t + ...$, and match coefficients of $\epsilon^n$ to obtain the result.

\end{proof}

\noindent The process $\Lambda^T$ 
 is the correction required for $\e^{Z+\Lambda^T}$ to be a martingale, with $\Lambda^T_T=0$ .  In particular, if $\e^Z$ is already a martingale, then 
$\Lambda^T \equiv 0$ on $[0,T]$.  
\begin{remark}
Note that the convergent $\mG$-sum is exactly equal to $\Lambda^T$, which satisfies the ``abstract Riccati'' equation~\eqref{eq:Master}, 
\[
  \Lambda_{t}^T =   \mathbb{E}_t Z_{t,T} + \tfrac{1}{2} \left( ( Z + \Lambda^T) \dm ( Z + \Lambda^T) \right)_t(T) .
\]

\end{remark}

To make the recursion~\eqref{eq:Grec} more concrete, consider the first few forests written diagrammatically
with \tkz$\A$ \tkz as a short-hand for $Y$, interpreted as single leaf:
\tkz \begin{eqnarray} \label{some_forests}
    \mG^2  &=&  (\tfrac{1} 2 a^2+b)\, \AAd \nonumber\\
\mG^3 &=& a\, (\tfrac 12 a^2+b)\, \AAdAd \nonumber\\
\mG^4 &=& \tfrac1{2} (\tfrac 12 a^2+b)^2\,\AAdAAdd+ 
a^2\, (\tfrac 12 a^2+b)\,  \AAdAdAd
\nonumber\\
\mG^5 
&=& a\, (\tfrac 12 a^2+b)^2\,\,\AAdAddAAd
+  \tfrac1{2} a\,(\tfrac 12 a^2+b)^2 \,\AAdAAddAd
+a^3\, (\tfrac 12 a^2+b)\,\,\AAdAdAdAd    \label{eq:Gforests}
      \end{eqnarray}  \tkz 

We note that when $b = - \tfrac12 a^2$, all of the corrector terms $\mG^k$ vanish, as they must, because in this case $\e^Z=\exp\left\{a Y - \tfrac12 a^2 \angl{Y}\right\}$ is the exponential martingale and $\Lambda=0$.  With $a=-b/2$, the $\mG$-recursion~\eqref{eq:Grec} becomes precisely the $\mF$-recursion, Equation (3.1) of~\cite{alos2020exponentiation} derived in the context of stochastic volatility modeling, whereas the case $b=0$ yields the $\mK$ (cumulant) recursion, Equations (3.4), (3.9) of~\cite{lacoin2022probabilistic}, derived in the context of renormalisation of the sine-Gordon model in quantum physics. 

In~\cite{friz2022forests}, a number of applications of the $\mK$ expansion are given, including a neat derivation of the moment generating function of the L\'evy area.  Other applications 
include 
the computation of likelihood functions in statistics,
the computation of correlation functions in statistical physics, and the
computation of amplitudes in quantum field theory. 

However, in the context of (rough) stochastic volatility modeling that originally gave rise to diamond expansions, the $\mG$-expansion and its special case, the $\mF$-expansion (with $b=0$), are the relevant ones.

\section{Model-free results under stochastic volatility} \label{sec:SV}

Consider a stochastic volatility model written in forward variance form.  Specifically, 
let~$S$ be a strictly positive continuous martingale (the stock price). Then 
$X := \log S / S_0$ 
is a continuous semimartingale with quadratic variation process $\angl{X}$, which is assumed 
continuously differentiable such as to have a well-defined 
 instantaneous variance, defined via
$$
    V_t\,\D t := \D \angl{X}_t \, ,
    $$ and then forward variance as the conditional expectation of future instantaneous variance
    \[
 \xi_t(T) = \eef{V_T} .
\]
Forward variances are tradable assets (unlike spot variance), 
constituting a family of martingales indexed by their individual time horizons $T$.

The fair strike of a (total) variance swap maturing at time $T$ is  given by 
\[
M_t(T)= \int_t^T\,\xi_t(u)\,\D u
\]
and the fair strike of a forward-starting variance swap, starting at time $T$ and maturing at time $T+\Delta$ by 
\[
\zeta_T(T) =  \int_T^{T+\Delta}\,\xi_T(u)\,\D u= \EE_T{\int_T^{T+\Delta}\,V_u\,\D u} = \EE_T \angl{X}_{T,T+\Delta}.
\]
Its price process $\zeta_t(T)  = \EE_t \angl{X}_{T,T+\Delta}$ for $t \in [0,T]$ defines a martingale $\zeta = \zeta(T)$. 

We now show how a straightforward application of Theorem~\ref{thm:mainintroG} gives a {\em model-free} expression for the moment generating function (mgf) of $\log S$, the variance swap, and the forward-starting variance swap. As for the practical interest, if
$S=S_0 \e^X$ represents the SPX index, and $\Delta$ is 30 days, we get the joint mgf of SPX, the variance swap and $VIX^2$.

\begin{corollary} \label{thm:TripleJointMGF} 
Under natural integrability conditions, for $a, b, c \in \RR$ sufficiently small,
\begin{equation} \label{generalCF}
\eef{\e^{a\,X_T+ b\,\angl{X}_{T} +c\,\zeta_T(T)}} 
=\exp\left\{a\, X_t + b\,\angl{X}_{t} +  c\,\zeta_t(T)+\sum_{k=2}^\infty\,\mG^k_t (T;a,b,c)  \right\} \;,
\end{equation}
with $0 \le t \le T$, where the $\mG^k$'s are given recursively by (\ref{eq:Grec}),  starting with 
\beq \label{generalCFmG2}
           \mG^2 = \left( \tfrac 12 a(a-1)  + b \right) X \dm X + a c \,X \dm \zeta + \tfrac{1}{2}c^2\,\zeta \dm \zeta \ .
\eeq
\end{corollary}
\begin{proof} 
Re-express the exponent at terminal time $T$ in terms of the martingale $Y = X + \tfrac{1}{2} \angl{X} $,
$$
  a\,X_T+ b\,\angl{X}_{T}+c\,\zeta_T(T) = a\,Y_T+c\,\zeta_T(T) +\left(b-\tfrac12a\right) \angl{Y}_{T}\, . 
$$
This is the inner product of $(a, c) \in \R^2$ with the $2$-dimensional martingale $(Y,\zeta) = (Y,\zeta(T)$, plus $\mathrm{diag} ( b - a/2, 0)$ contracted against its $(2 \times 2)$-covariation matrix, evaluated at time $T$. The bivariate formulation of Theorem~\ref{thm:mainintroG} now gives the right-hand side exponent
$$   a\,Y_t+c\,\zeta_t(T) +\left(b-\tfrac12a\right) \angl{Y}_{t} = a\,X_t+ b\,\angl{X}_{t}+c\,\zeta_t(T)
$$
plus the $\mG$-series, started with 
$$
         \mG^2 = \tfrac 12 a^2 Y \dm Y + ac Y \dm \zeta + \tfrac 12 c^2 \zeta \dm \zeta + (b - a/2) Y \dm Y.
$$
It remains to collect terms, noting that $Y$ can be safely replaced by $X$, since they only differ by a bounded variation term invisible to the diamond product. 
\end{proof}

\begin{remark} Martingality of $S=S_0 \e^X$ is seen in (\ref{generalCF}) upon setting $b=c=0$ and $a=1$. This is perfectly reflected in our $\mG$ expansion, which in this case, and also when $a=0$, vanishes altogether, as seen directly from~\eqref{generalCFmG2}. In this sense, martingality and total probability constraints are preserved at arbitrary truncation of the $\mG$ expansion of~\cite{friz2022forests} and its $\mF$-predecessor in~\cite{alos2020exponentiation}. We note that this is not achieved by the (cumulant) recursion of~\cite{lacoin2022probabilistic}: applied in its bivariate form, with $X = (1, -1/2) \cdot (Y,\langle Y \rangle)$, the martingality constraint is only seen upon summing a non-trivial infinite sum, whose value turns out to be zero. (In~\cite{friz2022forests} this is explained by forest reordering.)
\end{remark} 

\noindent Let $X\equiv \X$ and $\zeta \equiv \Zdotred$ so that 
we can write the first term of the $\mG$-series in Corollary~\ref{thm:TripleJointMGF} as
\tkz
\[
\mG^2 =\left( \tfrac{1}{2} a\,(a-1)+ b\right) \,\XXd +a c\, \XZd + \tfrac12 c^2\,\ZZd.
\]
\tkz
We could define \tkz$\XXd =: \M$\tkz, meaning $(X \dm X)=M = \M$, resulting in trees with leaves of three different colors, in which case $X_t$ would represents the log-stock price and $M_t(T)$, the variance swap.
Then
\beas 
 \mG^2 =  \left( \tfrac{1}{2} a\,(a-1)+ b\right) \,\M +a c\, \XZd + \tfrac12 c^2\,\ZZd.
\eeas
In general, we can always identify subtrees in this way and assign them a new variable name (and leaf color).

Upon setting $b=c=0$ in Corollary~\ref{thm:TripleJointMGF}, we get the $\mF$-expansion of~\cite{friz2022forests}. Upon Wick rotation, replacing~$a$ by~$\ui a$, and shifting the integer index, we recover the $\tilde{\mF}$-expansion of~\cite{alos2020exponentiation}.

 \begin{corollary}\label{cor:Ftilde}
The conditional cumulant  generating 
function (CGF) is given by
\beq
\psi_t(T; a) = \log \eet{\e^{\ui \,a \,X_{t,T}}} =  - \frac{1}{2}a\,(a+\ui)\,M_t(T) + \sum_{k=1}^\infty\,\tilde {\mF}_k(a).
\label{eq:CGFmF}
\eeq
where the $\tmF_k$ satisfy the recursion
$
\tmF_0 = -\tfrac 1 2 a(a+\ui)\,M_t = -\tfrac 1 2 a(a+\ui)\,\M
$
and for $k>0$,
\beq
\tmF_k=\frac12\,\sum_{j=0}^{k-2}\,\left(\tmF_{k-2-j} \dm \tmF_j\right) +\ui \,a\, \left(X \dm \tmF_{k-1}\right).
\label{eq:tmFrecursion}
\eeq
\end{corollary}

Having generated a model-free expression for the cumulant generating function for any stochastic volatility model written in forward variance form, we can compute model-free expressions for attainable claims, such as variance and gamma swaps. (See e.g.~\cite{lee2010gamma}.)

As is well-known, the fair value of the variance swap is given by the fair value of the log-strip:
\[
 \eef{X_{t,T}}= (-\ui)\,{\psi'_t}(T;0) =  - \tfrac{1}{2}\,M_t(T) 
\label{eq:vs}
\]
and the fair value of the gamma swap is given by the fair value of the entropy contract:
\[
\eef{X_{t,T} \,\e^{ X_{t,T}}} = -\ui\,{\psi'_t}(T;-\ui).
\]
From~\eqref{eq:CGFmF} and  the recursion~\eqref{eq:tmFrecursion}, it is easy to see that only trees containing a single 
$\X$
 leaf will survive in the sum after differentiation when $a = - \ui$ so that
\beas
\sum_{k=1}^\infty \,{\tilde{\mF}_k}'(-\ui) &=&\tfrac \ui 2\,\sum_{k=1}^\infty \,X^{\dm k} M\\
&=&\tfrac \ui 2\,\left\{ \MXd + \MXdXd+ \MXdXdXd +
...\right\}.
\eeas
Then the fair value of a gamma swap is given by
\beq
\cG_t(T) = 2\,\eef{X_{t,T} \,\e^{ X_{t,T}}} = \M + \MXd + \MXdXd+ \MXdXdXd + ...
\label{eq:gsAllOrders}
\eeq

\noindent We deduce that the fair value of a leverage swap is given by
\bea
\cL_t(T) &:=& \cG_t(T) - M_t(T) = \sum_{k=1}^\infty \,X^{\dm k} M \label{eq:leverage}\\
&=& \MXd + \MXdXd+ \MXdXdXd + ...\nonumber
\eea
The completely model-free and explicit expression~\eqref{eq:leverage} for the leverage swap is in terms of diamond products of products of the spot and volatility processes.  In particular, if spot and volatility processes are uncorrelated, $\MXd$ and higher order terms vanish, and we see that the fair value of the leverage swap is zero.

In~\cite{fukasawa2014volatility}, Fukasawa shows that
\beas
\eet{\int_t^T\,\D \angl{S,M(T)}_r} = S_t\left\{\cG_t(T) - M_t(T)\right\}.
\eeas
Thus, by definition of the diamond product and by the definition~\eqref{eq:leverage} of the leverage swap, 
\beq
 (S \dm M)_t(T)=S_t\,\cL_t(T).
\label{eq:levswap}
\eeq
In other words, the fair value of a leverage swap  is given by the quadratic covariation of the underlying and the variance swap,  {for  any admissible stochastic volatility model}.

Comparing equation~\eqref{eq:leverage} with equation~\eqref{eq:levswap},
we are led to the following identity, for which we offer a direct 
diamond proof. 

\begin{lemma}\label{lem:SdmM}

\[
(S \dm M)_t(T) = \left(\e^X \dm M\right)_t(T) = S_t\,\sum_{k=1}^\infty \,\left(X ^{k \dm} M \right)_t(T).
\]
\end{lemma}

\begin{proof}
Note that
\beas
&&(S \dm M)_t(T) = \eef{ \int_t^T \D \angl{S,M}_u}=\,\eef{ \int_t^T S_u\,\D \angl{X,M}_u}.
\eeas
Now define $U_t:=\eef{ \int_t^T \,\D \angl{X,M}_u}= (X \dm M)_t(T)$. Applying It\^{o}'s Formula to the process $S U$  we get
\begin{eqnarray*}
&&\eef{ \int_t^T S_u\,\D \angl{X,M}_u}\nonumber\\
&&=S_t\,\eef{ \int_t^T \,\D \angl{X,M}_u}+\eef{ \int_t^T S_u\,\D \angl{X,U}_u}\nonumber\\
&&=S_t\,(X\dm M)_t(T)+\eef{ \int_t^T  S_u\,\D \angl{X,(X \dm M)}_u}.
\end{eqnarray*}
Then a recursive argument gives us that
\[
\eef{ \int_t^T \,\D \angl{S,M}_u}=S_t\,\sum_{k=1}^\infty \,\left(X ^{k \dm} M \right)_t(T).
\]
\end{proof}

What is this good for?  Since the leverage swap is an attainable claim, it can be expressed as a portfolio of European options whose value is given in principle by the market -- only in principle, because in practice interpolation and extrapolation are required, market prices being available only certain discrete strikes and expirations.  Since~\eqref{eq:leverage} gives the model value of the leverage swap for any given maturity $T$, and we can estimate the value of the leverage swap for each $T$ from the market, models can be efficiently calibrated.

However there is a catch.  This model calibration procedure works only when we know how to compute or easily approximate diamond trees. 

\subsection{The Bergomi-Guyon smile expansion}\label{sec:chap_Forest_Chapter_BG}

Following~\cite{bergomi2012stochastic}, consider a forward variance model of the form
\bea
\frac {\D S_t}{S_t} &=& \sqrt{V_t}\,\D Z_t\nonumber\\
\D\xi_t(u) &=& \lambda(t,u,\xi_t)\,\D W_t.
\label{eq:Forest_BGform}
\eea

\noindent To expand such a model 
 for small volatility of volatility, Bergomi and Guyon scale the volatility of volatility function $\lambda(\cdot)$ by a dimensionless parameter $\varepsilon$ so that $\lambda \mapsto \varepsilon\,\lambda$.  Setting $\varepsilon=1$ at the end then gives the required expansion. 

Let $k=\log K/S$ be the log-strike and $T$ be the option expiration. For ease of notation
 and without loss of generality, set $t=0$ and drop references to the initial time.
 The Bergomi-Guyon smile expansion (Equation (14) of~\cite{bergomi2012stochastic}) then reads
\beq
\sigma_{\BS}(k,T) = \hat \sigma_T + \cS_T\,k + \cC_T\,k^2 + \cO(\varepsilon^3)
\label{eq:BGsmile}
\eeq
where the coefficients $\hat \sigma_T$, $ \cS_T$ and $\cC_T$  are expressed in terms of what Bergomi and Guyon call autocorrelation functionals.  The beauty of the Bergomi-Guyon smile expansion is that it relates observed properties of the implied volatility surface such as at-the-money volatility $\hat \sigma_T$ and the at-the-money volatility skew $\cS_T$ directly to quantities that are computable from the formulation of the model, for any model written in forward variance form. 

On the other hand the $\tilde{\mF}$-expansion~\eqref{eq:CGFmF} for the cgf can be expressed 
 as a formal power series in (our own parameter) $\epsilon$ whose power counts the forest index $\ell$.  That is, from~\eqref{eq:CGFmF},
\beq
\psi(T; a) = \log \ee{\e^{\ui \,a \,X_{T}}} =  - \frac{1}{2}a\,(a+\ui)\,M(T) + \sum_{\ell=1}^\infty\,\epsilon^\ell\,\tilde {\mF}_\ell(a).
\label{eq:CGFmF2}
\eeq


It is well-known~\cite{carr1999option, lewis2000option} that the price of a European option may be computed by inverting the characteristic function of a given stochastic process. Equivalently,  
equ. (5.7) of~\cite{gatheral2006volatility} expresses the  the implied total variance smile $\Sigma(k,T):= \sigma_{\BS}^2(k,T)\,T$ in 
terms of the characteristic function:
\begin{equation}
 \int_0^\infty\frac{\D u}{u^2+\frac{1}{4}}
\,\mathrm{Re}\left[\e^{-iuk}\left(\e^{\psi(T; u-i/2)}-\e^{-\frac{1}{2}\left(u^2+\frac{1}{4}\right)\,\Sigma(k,T)}\right)\right]=0.\label{eq:ImpliedFromFourier}
\end{equation}

\noindent Dropping explicit references to $T$, let
\[
\Sigma(k) =  \sum_{\ell=0}^\infty\,\epsilon^\ell\, a_\ell(k).
\]
Then, substituting from~\eqref{eq:CGFmF2},~\eqref{eq:ImpliedFromFourier} becomes
\beas
\int_0^\infty\frac{\D u}{u^2+\frac{1}{4}}
\,\mathrm{Re}\left[\e^{-iuk}\left\{\e^{ - \frac{1}{2}\left(u^2+\frac14\right)\,M(T) + \sum_{\ell=1}^\infty\,\epsilon^\ell\,\tilde {\mF}_\ell(u-\ui/2)}-\e^{-\frac{1}{2}\left(u^2+\frac{1}{4}\right)\,\sum_{\ell=0}^\infty\,\epsilon^\ell\, a_\ell(k)}\right\}\right]=0.
\eeas
Setting $\epsilon = 0$ gives $a_0(k) = M(T)$, as expected.  Then matching powers of $\epsilon$ to first order gives
\begin{equation} \label{eq:a1}
a_1(k) 
= \left(\frac{k}{M} + \frac12\right)\,(X \dm M).
\end{equation} 
where we have put $M= M(T)$ for short.   In Appendix A of~\cite{alos2020exponentiation}, an explicit algorithm using Hermite and Bell polynomials that matches coefficients of $\epsilon$ to any desired order is given. This algorithm gives the second order coefficient as
\beas
a_2(k) 
 &=&\frac 14  (X\dm M)^2\,\left\{-\frac{5 k^2}{M^3}-\frac{2 k}{M ^2} +  \frac{3}{M^2}+\frac{1}{4 M} \right\}\nonumber\\
&&+\frac 14 (M \dm M)\,   \left\{\frac{k^2}{M ^2}-\frac{1}{M}-\frac{1}{4}  \right\}+
(X \dm (X \dm M))\,  \left\{  \frac{k^2}{M ^2}+\frac{k}{M} -\frac{1}{M} +\frac{1}{4}    \right\} .\nonumber\\
  \eeas
It is straightforward to verify that the resulting expansion $\Sigma(k)= M + \epsilon\,a_1(k) + 
 \epsilon^2\,a_2(k) + \cO(\epsilon^3)$ agrees with the Bergomi Guyon smile expansion to second order in $\epsilon$, up to the identification $\epsilon \rightsquigarrow \varepsilon$.   In this sense, the $\tilde{\mF}$-expansion allows us to extend the Bergomi-Guyon expansion to all orders.

\section{Affine forward variance models}

Diamond trees turn out to be particularly easy to compute in affine forward variance (AFV) models.
As shown in Chapter~\ref{chap:RoughAffine_Chapter}, affine forward variance (AFV) models must take the form
\bea
\frac{\D S_t}{S_t}&=&\sqrt{V_t}\,\D Z_t\nonumber\\
\D\xi_t(u)&=&  \kappa(u-t)\,\sqrt{V_t}\,\D W_t,\qquad t\leq u,
\label{eq:AFV}
\eea
with \ $\D \angl{W,Z}_t = \rho\,\D t$. 
AFV models are therefore essentially all Heston models with different choices of the kernel function $\kappa$.  For example, the classical Heston model corresponds to the choice 
\[
\kappa(\tau) = \nu\,e^{-\lambda\,\tau},
\]
and the rough Heston model of~\cite{euch2019characteristic} to the choice
\[
\kappa(\tau) = \nu\,\tau^{\alpha-1}\,E_{\alpha,\alpha}(-\lambda\,\tau^\alpha),
\]
with $\nu,\lambda >0$ and $\alpha = H +\tfrac 12 \in \left(\tfrac12,1\right]$. 
Integrating the forward variance equation gives
\begin{equation} \label{eq:AFV2}
      \D \zeta_t (T) = \left( \int_T^{T+\Delta}  \kappa(u-t) \,\D  u \right) \,\sqrt{V_t}\,\D W_t =: \bar{\kappa}(T-t) \,\sqrt{V_t}\,\D W_t
\end{equation}
which has the same form as (\ref{eq:AFV}).

\begin{lemma}\label{lem:AffineTrees}

In the affine forward variance model (\ref{eq:AFV}) all diamond trees (with leaves of two types $X = \X$ and $\zeta = \Zdotred$, respectively), and hence all forests terms
$\mG^k_t$ in (\ref{generalCF}) are of the form 
\beq
\int_t^T\,\xi_t(u)\,h(T-u)\,\D u
\label{eq:convolutionForm}
\eeq
for some integrable function $h$.
\end{lemma}

\begin{proof} As we computed earlier, $\M = \XXd = \int_t^T\,\xi_t(s)\,\D s$, and from~\eqref{eq:AFV2},
\beq
\XZd \ \ \  
=  \rho\,\int_t^T\,\xi_t(u)\,\bar \kappa(T-u)\,\D u \;, 
\qquad \ZZd \ \ \  
=  \int_t^T\,\xi_t(u)\,\bar \kappa(T-u)^2\,\D u \;.
\label{eq:AFV3}
\eeq
We thus see that the claim holds for all diamond trees with two leaves and proceed by induction.
Consider two trees 
\[
\mT^i_t =\int_t^T\,\xi_t(u)\,h^i(T-u)\,\D u, \qquad i = 1,2
\]
of the supposed form. 
Then
\beas
\left(\mT^1 \dm \mT^2 \right)_t(T) &=& \eef{\int_t^T\,\D \angl{\mT^1,\mT^2}_u}\\
&=&\eef{\int_t^T \int_u^T \int_u^T\,\,h^1(T-s) \,h^2(T-r) \, \,\D s \, \D r \,\D \angl{\xi(s),\xi(r)}_u} \\
&=&\eef{\int_t^T\,V_u \,\kappa(s-u)\,\kappa(r-u)\,\D u\,\int_u^T\,h^1(T-s)\,\D s\,\int_u^T\,\,h^2(T-r)\,\D r\,}\\
&=&\int_t^T\,\xi_t(u)\,h^{12} (T-u) \,\D u,
\eeas
and the induction step is completed upon setting 
\[
h^{12}
(T-u)  = \int_u^T\,h^1(T-s)\,\kappa(T-s)\,\D s\,\int_u^T\,h^2(T-r)\,\kappa(T-r)\,\D r.
\]\end{proof}

At this stage it is tempting to combine Lemma~\ref{lem:AffineTrees} with Theorem~\ref{thm:TripleJointMGF} to compute the triple-joint mgf of $X_T$, $\angl{X}_{t,T}$, and $\zeta_T(T)$ by summing the full $\mG$-expansion for an affine forward variance model.  Since each tree in the $\mG$-expansion has the form $\int_t^T\,\xi_t(u)\,h(T-u)\,\D u$, it follows that the mgf is necessarily of the convolutional form
\[
\log \eef{\e^{a\,X_T + b\,\angl{X}_{T}+ c\,\zeta_T(T)}} =a\,X_t + b\,\angl{X}_{t} + c\,\zeta_t(T)+ \int_t^T \xi_t(u) \,g(T-u;a,b,c, \Delta)\, du,
\]
which amounts to an infinite-dimensional version of the classical affine ansatz.  
Inserting
$ \Lambda_{t} (T) = \int_t^T \xi_t(u) \,g(T-u;a,b,c, \Delta) \,\D u $
directly into the ``abstract Riccati'' equation~\eqref{eq:Master},  we 
 readily obtain that the triple-joint MGF satisfies a convolution Riccati equation of the type considered in~\cite{abijaber2019affine, gatheral2019affine}.
We summarise this in the following theorem.
\begin{theorem}\label{thm:triplejointMGF}
Let
\beas
\D X_t &=& -\tfrac 12 \,V_t\,\D t + \sqrt{V_t}\,\D Z_t\\
\D\xi_t(T) &=& \kappa(T-t)\,\sqrt{V_t}\,\D W_t,
\eeas
with $\D\angl{W,Z}_t=\rho\,\D t$ and let $\angl{X}_{t,T} = \angl{X}_T-\angl{X}_t$.  Further let $\tau = T-t$,
$
\bar \kappa(\tau) = \int_\tau^{\tau+\Delta}\,\kappa(u)\,\D u,
$
and define the convolution integral
$$
(\kappa \star g) (\tau) = \int_0^\tau\,\kappa(\tau - s)\,g(s)\,\D s.
$$
Then
\[
\eef{\e^{a\,X_T + b\,\angl{X}_{t,T}+ c\,\zeta_T(T)}} =\exp\left\{a\,X_t +c\,\zeta_t(T)+ (\xi \star g) (T-t;a,b,c,\Delta)\right\}
\]
where $g(\tau;a,b,c,\Delta)$ satisfies the convolution Riccati integral equation
\beq
g(\tau;a,b,c,\Delta)=b -\tfrac12 a + \tfrac12\,(1-\rho^2)\,a^2+ \tfrac12\,\left[\rho a + c \, \bar \kappa(\tau) + (\kappa \star g)(\tau;a,b,c,\Delta)\right]^2,
\label{eq:jointRiccati}
\eeq
with the boundary condition
$
g(0;a,b,c,\Delta)= b+\tfrac 12\,a(a-1) +\rho a c\,\bar \kappa (0)+ \tfrac12\,c^2\,\bar \kappa(0)^2
$.

\end{theorem}
\begin{proof}
From~\eqref{eq:AFV2},
$
d\zeta_t(T) 
= \sqrt{V_t}\,\D W_t\,{\bar \kappa}(T-t)
$.
As before, let 
$ \Lambda_{t}  = \int_t^T \xi_t(u) g(T-u;a,b,c, \Delta) du $.  
Then dropping the arguments $a,b,c,\Delta$ for ease of notation,
\beas
d\Lambda_t&=& -\xi_t(t)\, g(T-t)\,\D t + \int_t^{T}\,\D \xi_t(s)\, g(T-s)\,\D s\\
&=& -V_t\, g(T-t)\,\D t + \sqrt{V_t}\,\D W_t\,\int_t^{T}\,\kappa(s-t)\, g(T-s)\,\D s\\
&=& -V_t\, g(T-t)\,\D t + \sqrt{V_t}\,\D W_t\,(\kappa \star g)(T-t).
\eeas
We compute
\beas
\D \angl{X}_t
&=& V_t\,\D t\\
\D\angl{X,\zeta}_t &=& \rho\,V_t\, \bar \kappa(T-t)\,\D t\\
\D\angl{\zeta}_t &=&  V_t\,\bar \kappa(T-t)^2\,\D t\\
\D\angl{X,\Lambda}_t &=&  \rho\,V_t\,(\kappa\star g)(T-t)\, dt \\
\D\angl{\Lambda}_t &=& V_t\, \left[ (\kappa\star g)(T-t) \right]^2\,\D t\\
\D\angl{\zeta,\Lambda}_t &=& V_t\, \bar \kappa(T-t)\,(\kappa\star g)(T-t)\,\D t.
\eeas
Integrating these terms from $t$ to $T$, then taking a time-$t$ conditional expectation allows us to compute all diamond products in the ``abstract Riccati'' equation~\eqref{eq:Master}
\[
\Lambda  = \left(  \tfrac 12 a(a-1) + b\right) \XXd 
+ a c\,\XZd + a\,\X \dm \Lambda +  \tfrac{1}{2} \,\left[c\,\Zdotred +  \Lambda \right]^{\dm 2}
\]
to yield
\[
g(\tau)=\left(  \tfrac 12 a(a-1) + b\right) + \rho a c\,\bar \kappa(\tau) +a\,\rho\, (\kappa \star g)(\tau)
+\tfrac12\,
\left[ c \, \bar \kappa(\tau) + (\kappa \star g)(\tau)\right]^2,
\]
which upon rearrangement gives~\eqref{eq:jointRiccati}.
Finally, $(\kappa \star g)(0) =0$ gives the boundary condition
\beas
g(0) = b+\tfrac 12\,a(a-1) +\rho a c\,\bar \kappa (0)+ \tfrac12\,c^2\,\bar \kappa(0)^2.
\eeas

\end{proof}

\section{Explicit computations under rough Heston}

Explicit computations are easiest in the rough Heston model with $\lambda=0$.
In this case, with $\alpha  = H + 1/2 \in \left( 1/2, 1\right]$, 
\beas
\D X_t &=& -\tfrac12 V_t\,\D t +\sqrt{V_t}\,\D Z_t\\
\D\xi_t(u) &=& \frac{\nu}{\Gamma(\alpha)}\,(u-t)^{\alpha-1}\,\sqrt{V_t}\,\D W_t,
\eeas
with $\D\angl{W,Z} = \rho\,\D t$.  Starting with the simplest tree, the fair strike of the (total) variance swap is given by
\beas
M_t(T) = \M = \XXd = (X \dm X)_t(T) &=&\int_t^T\,\xi_t(u)\,\D u,
\eeas
as we noted earlier.  Thus, abbreviating bounded variation terms as `BV', we have
\beas
\D X_t &=& \sqrt{V_t}\,\D Z_t + \mathrm{BV}\\
\D M_t &=& \int_t^T\,\D \xi_t(u)\,\D u + \mathrm{BV}\\
& =& \frac{\nu}{\Gamma(\alpha)} \,\sqrt{V_t} \,\left(\int_t^T\,\frac{\D u}{(u-t)^\gamma}\right)\,\D W_t+ \mathrm{BV}\\
& =& \frac{\nu\,(T-t)^\alpha}{\Gamma(1+\alpha)} \,\sqrt{V_t} \,\D W_t+ \mathrm{BV}.
\eeas

\noindent We now proceed to compute the first few $\tilde{\mF}$ forests explicitly.  There is only one tree in the forest $\tmF_1$.
\beas
\XMd=(X \dm M)_t(T) &=& \eef{\int_t^T \D\angl{X,M}_s}\nonumber\\
&=& \frac{\rho\,\nu}{\Gamma(1+\alpha)} \,\eef{\int_t^T \,v_s \,(T-s)^\alpha\,\D s}\nonumber\\
&=& \frac{\rho\,\nu}{\Gamma(1+\alpha)} \,\int_t^T \,\xi_t(s) \,(T-s)^\alpha\,\D s.
\label{eq:rHcXM}
\eeas
In order to continue to higher orders, it makes sense to define for $j \geq 0$
\[
I^{(j)}_t(T):=\int_t^T\,\D s\,\xi_t(s)\,(T-s)^{j\,\alpha}.
\label{eq:Ij}
\]
Then
\beas
dI^{(j)}_s(T) &=& \int_s^T\,\D u\,\D \xi_s(u)\,(T-u)^{j\,\alpha} +{\mathrm{BV}}\nonumber\\
&=&\frac{\nu\,\sqrt{v_s}}{\Gamma(\alpha)}\,\D W_s\, \int_s^T\, \frac{(T-u)^{j\,\alpha}}{(u-s)^\gamma}\,\D u+{\mathrm{BV}} \nonumber\\
&=&\frac{ \Gamma (1+j \,\alpha)}{\Gamma (1+ (j+1)\, \alpha )}\,\nu\,\sqrt{v_s}\,(T-s)^{(j+1)\,\alpha}\,\D W_s+{\mathrm{BV}}.
\label{eq:dIj}
\eeas
With this notation, 
\[
(X \dm M)_t(T) = \frac{\rho\,\nu}{\Gamma(1+\alpha)} \,I^{(1)}_t(T).
\]

\noindent There are two trees in $\tmF_2$:
\beas
\MMd=(M \dm M)_t(T) &=& \eef{\int_t^T \D\angl{M,M}_s}\nonumber\\
&=& \frac{\nu^2}{\Gamma(1+\alpha)^2} \,\int_t^T \,\xi_t(s) \,(T-s)^{2\,\alpha}\,\D s\nonumber\\
&=& \frac{\nu^2}{\Gamma(1+\alpha)^2} \,I^{(2)}_t(T)
\label{eq:rHcMM}
\eeas
and
\beas
\MXdXd=\left(X \dm (X \dm M)\right)_t (T) &=& \frac{\rho\,\nu}{\Gamma(1+\alpha)} \, \eef{\int_t^T \D\angl{X,I^{(1)}}_s}\nonumber\\
&=& \frac{\rho^2\,\nu^2}{\Gamma(1+2\,\alpha)} \,I^{(2)}_t(T).
\label{eq:XXM}
\eeas
Note in passing that $\MMd$ gives the variance of the variance swap.
Continuing to the forest $\tmF_3$, we have the following.

\beas
\XMdMd=\left(M \dm (X \dm M)\right)_t (T) &=& { \frac{\rho\,\nu^3}{\Gamma(1+\alpha)\,\Gamma(1+2\,\alpha)} \,I^{(3)}_t (T)}\\
\MXdXdXd=\left(X \dm \left(X \dm (X \dm M)\right)\right)_t  (T) &=& \frac{\rho^3\,\nu^3}{\Gamma(1+3\,\alpha)} \,I^{(3)}_t (T)\\
\MMdXd=\left(X \dm (M \dm M)\right)_t  (T) &=& \frac{\rho\,\nu^3\,\Gamma(1+2\,\alpha)}{\Gamma(1+\alpha)^2\,\Gamma(1+3\,\alpha)} \,I^{(3)}_t (T).
\eeas

\noindent In particular, we readily identify the pattern
\beq\label{eq:levk}
\left(X ^{\dm k} M\right)_t (T)=  \frac{(\rho\,\nu)^k}{\Gamma (1+ k \,\alpha )}\,I^{(k)}_t(T),
\eeq
which gives us a simple closed-form expression for a tree with $k$ $\A$ leaves and only one $\M$ leaf.

\subsection{Time scaling of the $\tilde{\mF}$-expansion}

By inspection, fixing model parameters, we see that $\tilde{\mF}_k $ scales as
\[
I^{(k)}_t(T) =\int_t^T\,\D s\,\xi_t(s)\,(T-s)^{k\,\alpha} \sim (T-t) ^{k \alpha +1} \quad \text{ as } T \downarrow t.
\]
Thus, the $\tmF$-expansion of Corollary~\ref{cor:Ftilde}
has the interpretation of a short-time expansion, the concrete powers of which depend on the roughness parameter $\alpha = H + 1/2 \in (1/2,1)$. The resulting diamond expansions (which can obtained by alternative methods in the rough Heston case)  have been applied to construct efficient numerical schemes for the computation of the rough Heston characteristic function~\cite{callegaro2021fast, gatheral2019rational}.

\subsection{The leverage swap}

Recall that Equation~\eqref{eq:leverage} gives a model-free expression for the fair value of the leverage swap in terms of trees with only one $\M$ leaf.   Substituting from~\eqref{eq:levk}, we obtain
\beas
\cL_t(T) &=&  \sum_{k=1}^\infty \,\left(X ^{\dm k} M\right)_t(T)\nonumber\\
&=& \sum_{k=1}^\infty \,\frac{(\rho\,\nu)^k}{\Gamma (1+ k \,\alpha )}\,\int_t^T\,\D u\,\xi_t(u)\,(T-u)^{k\,\alpha}\nonumber\\
&=&\int_t^T\,\D u\,\xi_t(u)\,\left\{E_{\alpha }(\rho\,\nu\,(T-u)^\alpha)-1\right\}
\label{eq:rHestonLeverage}
\eeas
where $E_\alpha(\cdot)$ denotes the Mittag-Leffler function.  We have thus obtained an explicit expression for the leverage swap under rough Heston with $\lambda=0$.  As discussed earlier in Section~\ref{sec:SV}, since we can impute the leverage swap $\cL_t(t)$ from the smile for each expiration $T$, fast calibration of the rough Heston model is then possible.

\subsection{The rough Heston implied volatility skew}\label{sec:chap_Forest_Chapter_Heston}

As in Section~\ref{sec:chap_Forest_Chapter_BG}, fix $t=0$ and define $\Sigma(k,T)= \sigma_{\BS}^2(k,T)\,T$.  Then, from~\eqref{eq:a1}, in a general stochastic volatility model, the short  dated $(T \downarrow 0)$ total variance skew is given by
 \beas
\left. \p_k \Sigma(k,T) \right|_{k=0} = \frac{X \dm M}{M}.
 \eeas 
\noindent Recall that under rough Heston, $M_t(T) = \M = \XXd = (X \dm X)_t(T) = \int_t^T\,\xi_t(u)\,\D u$ and, 
$$
\XMd=(X \dm M)_t(T) = \frac{\rho\,\nu}{\Gamma(H + 3/2)} \,\int_t^T \,\xi_t(s) \,(T-s)^{H + 1/2}\,\D s.
$$
Setting $t=0$ in these expressions, with $\xi_0(s) \to \xi_0(0) = V_0$, time-$0$ spot variance as $s \downarrow 0$, an easy computation gives
$$
   \frac{\partial}{\partial k} \bigg|_{k=0} \sigma_{\mathrm{BS}}^2(k,T)  \sim  \frac{\rho\,\nu}{(H+3/2)\, \Gamma(H + 3/2)} T^{H-1/2} = \frac{\rho\,\nu}{\Gamma(H + 5/2)} T^{H-1/2}.
$$
As a sanity check, with $H=1/2$, we recover the classical Heston model with no mean reversion, and  noting that $\Gamma(3)=2!=2$, the well-known short-dated Heston implied variance skew formula, cf. \cite[p.35]{gatheral2006volatility}. For $H<1/2$ we see skew explosion of order $T^{H-1/2}$, as expected.

\section{Explicit computations under rough Bergomi}

The skeptical reader might wonder whether explicit diamond tree computations are only possible in AFV models.  We now show that explicit tree computations are also possible in the rough Bergomi model.

From~\cite{bayer2016pricing},
the rough Bergomi model may be written as
\beas
\frac{\D S_t}{S_t} &=& \sqrt{V_t}\,\left\{  \rho\,\D W_t + \sqrt{1-\rho^2}\,\D W^\perp_t \right\}\\
V_u &=& \xi_t(u) \,\cE\left(\eta\,\sqrt{2\,H}\, \int_t^u\,\frac{\D W_s}{(u-s)^\gamma} \right)
\eeas
with $\gamma = \frac 12 - H$.
In forward variance form
\[
\frac{\D_t \xi_t(u)}{\xi_t(u)} = \tilde \eta\,\frac{\D W_t}{(u-t)^\gamma} 
\]
with $\tilde \eta := \eta\,\sqrt{2\,H}$.
Then, again abbreviating bounded variation terms as `BV', we have
\beas
\D X_t &=& \sqrt{V_t}\,\D Z_t + \mathrm{BV}\\
\D M_t 
& =& \tilde \eta\,\left(\int_t^T\,\xi_t(u)\,\frac{\D u}{(u-t)^{\gamma}}\right)\,\D W_t +  \mathrm{BV}.
\eeas
\noindent  As in the rough Heston case, we now proceed to compute the first few $\tilde{\mF}$ forests explicitly.  
\subsection{The first order forest $\tmF_1$}
There is only one tree in the forest $\tmF_1$.
\beas
\MXd=(X \dm M)_t(T) &=& \eef{\int_t^T \D\angl{X,M}_s}\nonumber\\
&=& \rho\,\tilde \eta\,\eef{\int_t^T \,\D s\,\sqrt{v_s}\,\int_s^T\,\xi_s(u)\,\frac{\D u}{(u-s)^{\gamma}} }\nonumber\\
&=& \rho\,\tilde \eta\,\int_t^T \,\D s\,\int_s^T\,\eef{\sqrt{\xi_s(s)}\,\xi_s(u)}\,\frac{\D u}{(u-s)^{\gamma}} .
\label{eq:rBcXM}
\eeas
Applying Lemma~\ref{lem:prodSE}, we have
\beas
&&\eef{\sqrt{\xi_s(s)}\,\xi_s(u)}\\ &=& \sqrt{\xi_t(s)}\,\xi_t(u)\,\exp\left\{ \frac{\eta^2}2\,(s-t)^{2\,H}\,G_\gamma\,\left( \frac{u-t}{s-t}\right) \right\}
\, \exp\left\{-\frac{\eta^2}{8}\,(s-t)^{2\,H}\right\}.
\eeas
with $G_\gamma(\cdot)$ as defined in Appendix~\ref{sec:prodXiRoughBergomi}.  Thus
\bea
&&\MXd = (X \dm M)_t(T) \nonumber\\
&=& \rho\,\tilde \eta\,\int_t^T\,\D s\,\sqrt{\xi_t(s)}\,\int_s^T\,\frac{\D u}{(u-s)^\gamma}\,\xi_t(u)\,\exp\left\{ \frac{ {\eta}^2}2\,(s-t)^{2\,H}\,\left[ G_\gamma\left(\frac{u-t}{s-t}\right) - \frac 1{4} \right]\right\}.\nonumber\\
\label{eq:XdmM}
\eea

\subsection{The second order forest $\tmF_2$}\label{sec:F2}

\noindent There are two trees in $\tmF_2$, $\MMd=(M \dm M)$ and $\MXdXd=\left(X \dm (X \dm M)\right)$.  
First we compute
\beas
\MMd= (M \dm M)_t(T) &=& \eef{\int_t^T \D\angl{M,M}_s}\nonumber\\
&=&{\tilde \eta}^2\,\eef{\int_t^T \,\D s\,\int_s^T\,\xi_s(r)\,\frac{\D r}{(r-s)^{\gamma}}\,\int_s^T\,\xi_s(u)\,\frac{\D u}{(u-s)^{\gamma}}}\\
&=&2\,{\tilde \eta}^2\,\eef{\int_t^T \,\D s\,\int_s^T\,\xi_s(r)\,\frac{\D r}{(r-s)^{\gamma}}\,\int_r^T\,\xi_s(u)\,\frac{\D u}{(u-s)^{\gamma}}}.
\eeas
Another application of Lemma~\ref{lem:prodSE} gives, for $u \geq r$,
\beas
&&\eef{\xi_s(r)\,\xi_s(u)}\\ 
&=& \xi_t(r)\,\xi_t(u)\,\exp\left\{ \eta^2\,(s-t)^{2\,H}\,G_\gamma\,\left( \frac{u-t}{s-t},\frac{r-t}{s-t}\right) \right\}.
\eeas
Thus
\bea
&&\MMd= (M \dm M)_t(T)\nonumber\\ 
&=& 2\,{\tilde \eta}^2\,\int_t^T \,\D s\,\int_s^T\,\xi_t(r)\,\frac{\D r}{(r-s)^{\gamma}}\,\int_r^T\,\xi_t(u)\,\frac{\D u}{(u-s)^{\gamma}}
\,\exp\left\{ {\eta^2}\,(s-t)^{2\,H}\,G_\gamma\,\left( \frac{u-t}{s-t},\frac{r-t}{s-t}\right) \right\}.\nonumber\\
\label{eq:MdmM}
\eea
Reflecting the fact that $\MMd$ is the variance of $\M$, this integral may be simplified as 
\bea
&&\MMd = (M \dm M)_t(T)\nonumber\\ 
&=&2\, \int_t^T\,\xi_t(u)\,\D u\,\int_t^u\,\xi_t(r)\,\D r\,\left[ 
\exp\left\{  \eta^2\,(r-t)^{2 H}\,G_\gamma\left(\frac{u-t}{r-t}\right) \right\}-1 \right].
\label{eq:cMM}
\eea

\noindent Next we compute
\[
\MXdXd=\left(X \dm (X \dm M)\right)_t (T) =  \eef{\int_t^T \D\angl{X,(X\dm M)}_s}.
\]
To compute this, we need that
\[
\D (X \dm M)_s = \rho\,\tilde \eta\,\left\{ \D I_s+ \D J_s\right\}
\]
where 
\beas
\D I_s &=& \frac12\, \int_s^T\,\D r\,\frac{d \xi_s(r)}{\sqrt{\xi_s(r)}}\,\int_r^T\,\frac{\D u}{(u-r)^\gamma}\,\xi_s(u)\,\exp\left\{ \frac{ {\eta}^2}2\,(r-s)^{2\,H}\,\left[ G_\gamma\left(\frac{u-s}{r-s}\right) - \frac 1{4} \right]\right\}\\
&=& \D W_s\, \frac{ \tilde \eta}2\, \int_s^T\,\frac{\D r}{(r-s)^\gamma}\,\sqrt{\xi_s(r)}\,\int_r^T\,\frac{\D u}{(u-r)^\gamma}\,\xi_s(u)\,\exp\left\{ \frac{ {\eta}^2}2\,(r-s)^{2\,H}\,\left[ G_\gamma\left(\frac{u-s}{r-s}\right) - \frac 1{4} \right]\right\}
\eeas
and
\beas
\D J_s &=&  \int_s^T\,\D r\,\sqrt{\xi_s(r)}\,\int_r^T\,\frac{\D u}{(u-r)^\gamma}\,\D \xi_s(u)\,\exp\left\{ \frac{ {\eta}^2}2\,(r-s)^{2\,H}\,\left[ G_\gamma\left(\frac{u-s}{r-s}\right) - \frac 1{4} \right]\right\}\\
&=& \D W_s\, \tilde \eta\, \int_s^T\,\D r\,\sqrt{\xi_s(r)}\,\int_r^T\,\xi_s(u)\,\frac{\D u}{(u-s)^\gamma\,(u-r)^\gamma}\,\exp\left\{ \frac{ {\eta}^2}2\,(r-s)^{2\,H}\,\left[ G_\gamma\left(\frac{u-s}{r-s}\right) - \frac 1{4} \right]\right\}.
\eeas

\noindent Also from Lemma~\ref{lem:prodSE} again,
\beas
&&\eef{\sqrt{\xi_s(s)}\,\sqrt{\xi_s(r)}\,\xi_s(u)}\\
&=&\sqrt{\xi_t(s)}\,\sqrt{\xi_t(r)}\,\xi_t(u) \,\\
&& \qquad \qquad \times \exp\left\{ \eta^2\,(s-t)^{2\,H}\, \left[ \frac14 G_\gamma\left(\frac{r-t}{s-t}\right)+ \frac12 G_\gamma\left(\frac{u-t}{s-t}\right)
+ \frac12 G_\gamma\left(\frac{u-t}{s-t},\frac{r-t}{s-t}\right)\right]\right\}  \\
&& \qquad \qquad \qquad \qquad \times \exp\left\{ -\frac18\,\eta^2\,\left[ (s-t)^{2\,H} +  (r-t)^{2\,H} -  (r-s)^{2\,H}  \right]\right\}.
\eeas
Define
\beas
&&\eta^2\,F(s-t,r-t,u-t)\\ &=& \eta^2\,(s-t)^{2\,H}\, \left[ \frac14 G_\gamma\left(\frac{r-t}{s-t}\right)+ \frac12 G_\gamma\left(\frac{u-t}{s-t}\right)
+ \frac12 G_\gamma\left(\frac{u-t}{s-t},\frac{r-t}{s-t}\right)\right]\\
&&\qquad\qquad\qquad\qquad+ \frac{\eta^2}{2}\,(r-s)^{2\,H}\, G_\gamma\left(\frac{u-s}{r-s}\right) -\frac18\,\eta^2\,\left[ (s-t)^{2\,H} +  (r-t)^{2\,H}  \right].
\eeas
Then
\[
(X\dm (X \dm M))_t(T) = (\rho\,\tilde \eta)^2\,\left\{\frac12\, I_t + J_t\right\}
\]
with
\bea
I_t
&=& \int_t^T\,\D s\,\,\sqrt{\xi_t(s)} \int_s^T\,\frac{\D r}{(r-s)^\gamma}\,\sqrt{\xi_t(r)}\,\int_r^T\,\frac{\D u}{(u-r)^\gamma}\,\xi_t(u)\,\exp\left\{\eta^2\,F(s-t,r-t,u-t)\right\}\nonumber\\
\label{eq:J1}
\eea
and
\bea
J_t
&=&\int_t^T\,\D s\,\sqrt{\xi_t(s)} \int_s^T\,\D r\,\sqrt{\xi_t(r)}\,\int_r^T\,\frac{\D u}{(u-s)^\gamma\,(u-r)^\gamma}\,\xi_t(u)\,\exp\left\{\eta^2\,F(s-t,r-t,u-t)\right\}.\nonumber\\
\label{eq:J2}
\eea

One might ask whether it is possible to simplify these expressions using integration by parts as we did for $M \dm M$.  Unfortunately, it seems that we cannot.

In summary, though computations of diamond trees under rough Bergomi are possible in closed-form as shown above, computations get more and more complicated for higher order forests and numerical implementation is far from straightforward.  To make further progress, smart approximations will doubtless be required.

\bibliographystyle{alpha}

\bibliography{RoughVolatility}

\begin{appendix}
\section{Conditional expectations of products of forward variances under rough Bergomi}\label{sec:prodXiRoughBergomi}

\noindent We will need the following Lemma which is obtained by explicit integration.

\begin{lemma}\label{lem:Ggamma}
Let $u_j>u_i \geq s>t$ and $\gamma = \frac 12-H$. Then
\[
\int_t^{s}\,\frac{\D r}{(u_j-r)^\gamma\,(u_i-r)^\gamma} = \frac{1}{2\,H}\,(s-t)^{2\,H}\,G_\gamma\,\left(\frac{u_j-t}{s-t}, \frac{u_i-t}{s-t}\right) 
\]
where for $y\geq x \geq 1$,
\bea
G_\gamma(y,x)&=&2\,H\,\int_0^1\,\frac{\D r}{(y-r)^\gamma\,(x-r)^\gamma}\nonumber\\
&=& \frac{1-2\,\gamma}{(1-\gamma) \,  (y-x)}\,\bigg\{
{x^{1-\gamma } y^{1-\gamma
   } \, _2F_1\left(1,2-2 \gamma ;2-\gamma ;\frac{y}{y-x}\right)}\nonumber\\
   && \quad\quad\quad
   - {(x-1)^{1-\gamma } (y-1)^{1-\gamma } \, _2F_1\left(1,2-2 \gamma ;2-\gamma
   ;\frac{y-1}{y-x}\right)}
   \bigg\}.\nonumber\\
   \label{eq:Ggamma}
\eea
Equation~\eqref{eq:Ggamma} has the special cases
$$
G_\gamma(y):= G_\gamma(y,1) =\frac{1-2\,\gamma}{(1-\gamma) \,  (y-1)}\,{ y^{1-\gamma
   } \, _2F_1\left(1,2-2 \gamma ;2-\gamma ;\frac{y}{y-1}\right)}
$$ 
with $y>1$.  In particular
$
G_\gamma(1) = G_\gamma(1,1) =1
$.
\end{lemma}

\begin{lemma} \label{lem:prodSE}
Consider the Rough Bergomi model
\[
V_u = \xi_u(u) = \xi_t(u)\,\cEE{\tilde \eta\,\int_t^u\,\frac{\D W_r}{(u-r)^\gamma}}
\]
and let $u_n>...>u_i>u_{i-1}>...\geq s>t$, $i= 1,...n$. Then,
\beas
&&\eef{\prod_{i=1}^n{\xi_s(u_i)}^{\alpha_i}}\\
&=& \left(\prod_{i=1}^n{\xi_t(u_i)}^{\alpha_i}\right)\,
\exp\left\{ \eta^2\,(s-t)^{2\,H}\,\sum_{j>i}\alpha_i\,\alpha_j\,G_\gamma\,\left(\frac{u_j-t}{s-t}, \frac{u_i-t}{s-t}\right) \right\}\\
&&\quad\quad\quad\quad\quad\quad\quad\quad\times \exp\left\{\frac{1}{2}\,\eta^2\,\sum_{i=1}^n\,\alpha_i\,(\alpha_i-1)\,\left[(u_i-t)^{2\,H}-(u_i-s)^{2\,H}\right]\,\right\}.
\eeas
\end{lemma}

\begin{proof}
Let $X$ be a Gaussian random variable and $\alpha \in \mathbb{R}$.  Then
\[
\cEE{X}^\alpha = \cEE{\alpha\,X} \,\exp\left\{\frac{\alpha\,(\alpha-1)}{2}\,\var[X]\right\}.
\]
Also, with $\gamma = \frac 12 - H$, $\tilde \eta = \eta\,\sqrt{2\,H}$ and $u > s >t$,
\[
\xi_s(u) = \xi_t(u)\,\cEE{\tilde \eta\,\int_t^s\,\frac{\D W_r}{(u-r)^\gamma}}
\]
so
\[
\xi_s(u)^\alpha = \xi_t(u)^\alpha\,\cEE{\alpha\,\tilde \eta\,\int_t^s\,\frac{\D W_r}{(u-r)^\gamma}}\,\exp\left\{\frac12\,\eta^2\,\,\alpha\,(\alpha-1)\,\left[(u-t)^{2\,H}-(u-s)^{2\,H}\right]\right\}.
\]
Also, if $X_i, i=1,...,n$ are zero mean Gaussian random variables, then
\[
\prod_{i=1}^n \cE(X_i) = \cE\left(\sum_{i=1}^n\,X_i\right)\,\exp\left\{\sum_{j>i}^n\,\cov[X_i\,X_j]\right\}.
\]
Finally, assuming wlog that $u_j > u_i$, we have from Lemma~\ref{lem:Ggamma} that
\beas
\eef{\alpha_j\,\tilde \eta\,\int_t^s\,\frac{\D W_r}{(u_j-r)^\gamma} \,\alpha_i\,\tilde \eta\,\int_t^s\,\frac{\D W_{r'}}{(u_i-{r'})^\gamma}}
&=&   \alpha_i\,\alpha_j\, {\tilde \eta}^2\,  \int_t^s\,\frac{\D r}{(u_j-r)^\gamma\,(u_i-r)^\gamma} \\
&=& \alpha_i\,\alpha_j\, \eta^2\,  (s-t)^{2\,H}\,G_\gamma\,\left(\frac{u_j-t}{s-t}, \frac{u_i-t}{s-t}\right).
\eeas
\end{proof}

\end{appendix}

\end{document}